\documentclass[a4paper,UKenglish,autoref]{lipics-v2019}

\usepackage[ruled,vlined,linesnumbered,commentsnumbered]{algorithm2e}
\usepackage{xcolor}
\usepackage{hyperref}

\title{Improved approximation algorithms for some capacitated {\em k} edge connectivity problems}

\titlerunning{Improved approximation algorithms for some capacitated {\em k} edge connectivity problems}

\author{Zeev Nutov}{The Open University of Israel}{nutov@openu.ac.il}
{https://orcid.org/0000-0002-6629-3243}{}
\authorrunning{Zeev Nutov}
\Copyright{Zeev Nutov}

\nolinenumbers


\acknowledgements{}

\begin{document}

\maketitle
\newcommand {\ignore} [1] {}

\newcommand{\sem}    {\setminus}
\newcommand{\subs}   {\subseteq}
\newcommand{\empt}  {\emptyset}

\newcommand{\f}   {\frac}

\def\A {\mathbb{A}}

\def\al   {\alpha}
\def\be {\beta}
\def\ga {\gamma}
\def\de   {\delta}
\def\eps {\epsilon}
\def\la {\lambda}

\def\CC  {{\cal C}}
\def\FF  {{\cal F}}
\def\LL  {{\cal L}}
\def\PP  {{\cal P}}


\def\csnd  {{\sc Cap-SND}}

\def\sfec     {{\sc Set Family Edge Cover}}
\def\nmcc  {{\sc Near Min-Cuts Cover}}

\def\ckecs    {{\sc Cap-$k$-ECS}}                            
\def\ckecsa {{\sc Cap-$k$-ECS Augmentation}} 

\def\fgc        {{\sc FGC}}                                           
\def\fgca      {{\sc FGC Augmentation}}                  

\def\stfgc   {{\sc $st$-FC}}                                  
\def\stfgca  {{\sc $st$-FC Augmentation}}         

\keywords{capacitated network design, flexible connectivity, near minimum cuts}

\begin{abstract}
We consider the following two variants of the {\sc Capacitated $k$-Edge Connected Subgraph} ({\ckecs}) problem.
\begin{itemize}
\item
{\nmcc}: 
Given a graph $G=(V,E)$ with edge costs and $E_0 \subs E$, find a min-cost edge set $J \subs E \sem E_0$ 
that covers all cuts with at most $k-1$ edges of  the graph $G_0=(V,E_0)$. 
We obtain approximation ratio $k-\la_0+1+\eps$, improving the ratio $2\min\{k-\la_0,8\}$ of \cite{BCGI} for $k-\la_0 \leq 14$,
where $\la_0$ is the edge connectivity of $G_0$.
\item
{\sc $(k,q)$-Flexible Graph Connectivity} ($(k,q)$-{\fgc}): 
Given a graph $G=(V,E)$ with edge costs and a set $U \subs E$ of ''unsafe'' edges and integers $k,q$, 
find a min-cost subgraph $H$ of $G$ such that every cut of $H$ has at least $k$ safe edges or at least $k+q$ edges. 
We show that 
$(k,1)$-{\fgc} admits approximation ratio $3.5+\eps$ if $k$ is odd (improving the ratio $4$ of \cite{BCHI}), and that 
$(k,2)$-{\fgc} admits approximation ratio $6$ if $k$ is even and $7+\eps$ if $k$ is odd (improving the ratio $20$ of \cite{BCGI}).
\end{itemize}
\end{abstract}

\section{Introduction} \label{s:intro}

Let $G=(V,E)$ be  graph. 
For an edge subset or a subgraph $J$ of $G$ and $S \subs V$ let $\de_J(S)$ denote 
the set of edges in $J$ with one end in $S$ and the other in $V \sem S$, 
and let $d_J(S)=|\de_J(S)|$ be their number;
we say that {\bf $J$ covers $S$} if $d_J(S) \geq 1$.
We say that $J$ covers a set family $\FF$ if $J$ covers every set $S \in \FF$.
For a proper node subset $S$ of a graph $H$, the {\bf cut} defined by $S$ is $\de_H(S)$
(it is known that if $H$ is connected then distinct sets define distinct cuts);  
an {\bf $\ell$-cut} means that $d_H(S)=\ell$. 
We consider several variants of the following problem.  

\begin{center}
\fbox{\begin{minipage}{0.96\textwidth} \noindent
\underline{{\sc Capacitated $k$-Edge Connected Subgraph} ({\ckecs})} \\ 
{\em Input:} \ \ A graph $G=(V,E)$ with edge costs $\{c_e:e \in E\}$ and edge capacities $\{u_e:e \in E\}$, 
\hphantom {\em Input: \ } and an integer $k$. (All input  numbers are assumed to be non-negative integers). \\ 
{\em Output:} A mini-cost edge set $J \subs E$ such that $u(\de_J(S)) \geq k$ for every $\empt \neq S \subset V$.
\end{minipage}} \end{center}

In the augmentation version {\ckecsa} of {\ckecs} we are also given a subgraph $G_0=(V,E_0)$ of $G$ 
of cost zero that has capacitated edge connectivity $\la_0=\min\{u(\de_{E_0}(S)):\empt \neq S \subsetneq V\}$ close to $k$. 
The goal is to compute a min-cost edge set $J \subs E \sem E_0$ such that $G_0 \cup J$ has capaciatated connectivity $k$. 
For $k=\la_0-1$ this problem is equivalent to the ordinary {\sc $k$-Edge Connectivity Augmentation} problem, where all capacities are $1$. 
We consider a specific version of {\ckecsa} when 
every edge in $E \sem E_0$ has capacity $\geq k -\la_0$, and w.l.o.g. assume that edges in $E_0$ have capacity $1$.
Consequently, this problem can also be formulated as follows. 

\begin{center}
\fbox{\begin{minipage}{0.96\textwidth} \noindent
\underline{\nmcc} \\ 
{\em Input:} \ \ A graph $G=(V,E)$, $E_0 \subs E$ with edge costs $\{c_e:e \in E \sem E_0\}$, and an integer $k$. \\ 
{\em Output:} A min-cost edge set $J \subs E \sem E_0$ that covers the family $\{\empt \neq S \subsetneq V: d_{E_0}(S) < k\}$.
\end{minipage}} \end{center}

Using the $2$-approximation algorithm of \cite{GGPS} for increasing the edge connectivity by $1$, 
one can easily obtain ratio $2(k-\la_0)$ for this problem. 
Using a more recent $(1.5+\eps)$-approximation algorithm of \cite{TZ,TZ2} gives ratio $(1.5+\eps)(k-\la_0)$.
Bansal, Cheriyan, Grout, and Ibrahimpur \cite{BCGI} showed that {\nmcc} admits a constant approximation ratio $16$. 
We improve over these ratios for the cases when $\la_0$ is close to $k$. 

\begin{theorem} \label{t:la}
{\nmcc} admits the following approximation ratios:
\begin{itemize}
\item
$k-\la_0$                   \hspace*{1.55cm} if $\la_0,k$ are both even. 
\item
$k-\la_0+1/2+\eps$ \                                if $\la_0,k$ have distinct parity.
\item
$k-\la_0+1+\eps$    \hspace*{0.38cm}  if $\la_0,k$ are both odd.
\end{itemize}
\end{theorem}

\noindent
For example, when $k-\la_0 =8$ and $\la_0,k-1$ are both even, our ratio is $8$ while that of \cite{BCGI} is $16$. 
On the other hand, when $k-\la_0 \geq16$ the ratio $16$ of \cite{BCGI} is better than our ratio. 

\medskip

Recently, Adjiashvili, Hommelsheim and M\"{u}hlenthaler \cite{AHM} 
defined a new interesting version of {\ckecs}, called $k$-{\sc Flexible Graph Connectivity} ($k$-{\fgc}).
Suppose that there is a subset $U \subs E$ of ``unsafe'' edges, and we want to find 
the cheapest spanning subgraph $H$ that will be $k$-connected even if some unsafe edge is removed. 
This means that for any proper subset $S$ of $V$ we should have $d_{H \sem U}(S) \geq k$ or $d_H(S) \geq k+1$.
It is not hard to see that $k$-{\fgc} is equivalent to {\sc Cap-$(k(k+1))$-ECS} with $u(e)=k$ if $e \in U$ and $u(e)=k+1$ otherwise. 
Furthermore, for $U=\empt$ we get the ordinary {\sc $k$-Edge-Connected Subgraph} problem. 

Boyd, Cheriyan, Haddadan, Ibrahimpur \cite{BCHI} suggested a generalization of {\fgc} 
when up to $q$ unsafe edges may fail.
Let us say that a subgraph $H=(V,J)$ of $G$ is {\bf $(k,q)$-flex-connected} if
any cut $\de_H(S)$ of $H$ has at least $k$ safe edges or at least $k+q$ (safe and unsafe) edges, namely, if  
$d_{H \sem U}(S) \geq k$ or $d_H(S) \geq k+q$ for all $\empt \neq S \subsetneq V$.
Observing that $d_{H \sem U}(S) =d_H(S)-d_{H \cap U}(S)$,  
we get that $H$ is $(k,q)$-flex-connected if and only if
\begin{equation} \label{e:flex}
d_H(S) \geq k+\min\{d_{H \cap U}(S),q\} \ \ \ \ \ \forall \empt \neq S \subsetneq V
\end{equation} 
Summarizing, we get the following problem.

\begin{center}
\fbox{\begin{minipage}{0.96\textwidth} \noindent
\underline{{\sc $(k,q)$-Flexible Graph Connectivity} ($(k,q)$-{\fgc})} \\ 
{\em Input:} \ \ A graph $G=(V,E)$ with edge costs $\{c_e:e \in E\}$, $U \subs E$, and integers $k,q \geq 0$. \\  
{\em Output:} A min-cost subgraph $H$ of $G$ such that (\ref{e:flex}) holds.
\end{minipage}} \end{center}

As was mentioned, $(k,1)$-{\fgc} reduces to {\sc Cap-$(k(k+1))$-ECS} with $u(e)=k$ if $e \in U$ and $u(e)=k+1$ otherwise. 
It is also known (see \cite{CJ}) that $(1,q)$-{\fgc} reduces to {\sc Cap-$(q+1)$-ECS} with $u(e)=1$ if $e \in U$ and $u(e)=q+1$ otherwise. 
However, no reduction to {\sc Cap-$\ell$-ECS}  is known for other values of $k,q$.

The best ratio for $(k,1)$-{\fgc} was $4$ for arbitrary costs and $16/11$ for unit costs \cite{BCHI}.
The best ratio for $(k,2)$-{\fgc} was $2k+4$ for $k \leq 7$ \cite{CJ} and $20$ for $k \geq 8$ \cite{BCGI}. 
We improve this as follows, and also give a simple approximation algorithm for $(k,q)$-{\fgc} with unit costs. 

\begin{theorem} \label{t:fgc}
\begin{enumerate}[(i)]
\item
$(k,1)$-{\fgc} admits approximation ratio $3.5+\eps$ if $k$ is odd.
\item
$(k,2)$-{\fgc} admits approximation ratio $7+\eps$ if $k$ is odd and $6$ if $k$ is even.
\item
For unit costs, $(k,q)$-{\fgc} admits approximation ratio $\al+\f{2q}{k}$, where $\al$ is the best known ratio for the 
{\sc Min-Size $k$-Edge-Connected Subgraph} problem.
\end{enumerate}
\end{theorem}

We summarize the best known approximation ratios for $(q,k)$-{\fgc} in Table~\ref{tbl:fgc}. 
\begin{table} [htbp] 
\begin{center}
\begin{tabular}{|c|l|l|}  \hline  
{\boldmath $(k,q)$}  & {\bf previous}      & {\bf this paper}              
\\\hline \hline
$(k,1)$                      & $4$ \cite{BCHI}                                                                                  & $3.5+\eps$ for $k$ odd    \\
                                  & $16/11$ for unit costs \cite{BCHI}                                                    & $\al+2/k$ for unit costs 
\\\hline 
$(k,2)$                      & $2k+4$ \cite{CJ} for $k \leq 7$, $20$ for $k \geq 8$ \cite{BCGI} & $7+\eps$ for $k$ odd, $6$ for $k$ even
\\\hline 
$(k,3)$                      & $4k+4$ \cite{CJ}                                                                               &  
\\\hline
$(k,4)$                      & $6k+4$ for $k$ even  \cite{CJ}                                                        &        
\\\hline \hline
$(1,q)$                      & $q+1$ \cite{BCHI}                                                                             &  
\\\hline
$(2,q)$                      & $2q+2$    \cite{CJ}                                                                           & 
\\\hline
$(k,q)$                      & $O(q \log n)$   \cite{BCHI}                                                               &  $\al+\f{2q}{k}$ for unit costs
\\\hline
\end{tabular}
\end{center}
\caption{Known approximation ratios for $(k,q)$-{\fgc}.
Here $\al$ is the best known ratio for the 
{\sc Min-Size $k$-Edge-Connected Subgraph} problem; currently 
$\al=1+{1}/(2k)+O({1}/{k^2})$ for graphs \cite{GG}, and for multigraphs 
$\al=1+\f{3}{k}$ for $k$ odd and $\al=1+\f{2}{k}$ for $k$ even \cite{GGTW}, and $\al=1.326$ for $k=2$~\cite{GGJ}.}
\label{tbl:fgc}
\end{table}

In the next section \ref{s:pre} we give some uncrossing properties of minimum and minimum+1 cuts of a graph needed 
for the proofs of Theorems \ref{t:la} and \ref{t:fgc}, that are proved in Sections \ref{s:la} and \ref{s:fgc}, respectively.

\section{Uncrossing properties of minimum and minimum+1 cuts} \label{s:pre}

In this section we give some ``uncrossing'' properties of near minimum cuts 
needed for the proofs of Theorems \ref{t:la} and \ref{t:fgc}. 
Let $H=(V,J)$ be a $\la$-connected graph and let $A,B \subs V$ such 
that all four sets $C_1=A \cap B, C_2=A \sem B, C_3=V \sem (A \cup B), C_4=B \sem A$ are non-empty.
Shrinking each of these sets into a single node  results in a graph on $4$ nodes, 
and we will further replace parallel edges by a single capacitated edge; see Fig.~\ref{f:l-l+1}(a),
where near each edge is written its (currently unknown) capacity.
We will call this graph the {\bf square} of $A,B$, the sets that correspond to shrunken nodes are {\bf corner sets}, 
edges $C_1C_2,C_2C_3,C_3C_4,C_4C_1$ are {\bf side edges} (that have capacity $z,y,w,x$, respectively),
and $C_2C_4,C_1C_3$ are {\bf diagonal edges} (that have capacity $a,b$, respectively).
We will also use abbreviated notation $d_i=d_H(C_i)$, $d_H(A)=d_A$, $d_H(B)=d_B$.
The following equalities are known, and can be easily verified 
by counting the contribution of each edge to both sides of each equality:
\begin{eqnarray}
d_1+d_3 & = & d_A+d_B-2a  \label{e:a} \ , \\
d_2+d_4 & = & d_A+d_B-2b  \label{e:b} \ .
\end{eqnarray}

Since the function $d(\cdot)$ is {\bf symmetric}, namely, $d(S)=d(V \sem S)$ for all $S \subs V$, we will assume w.l.o.g that:
\begin{itemize}
\item
$d_1 \leq d_i$ for $i=2,3,4$.
\item
$d_2 \leq d_4$.
\item
If $d_1=d_2$ then $a \geq b$. 
\end{itemize}

\begin{figure} \centering \includegraphics{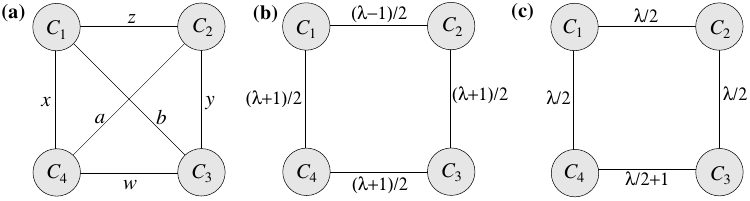}
\caption{Illustration to Lemmas \ref{l:sol} and \ref{l:l-l+1}.}
\label{f:l-l+1} \end{figure}

\begin{lemma} \label{l:sol}
Let $\al=d_A+d_1-d_2$. Then $\al$ is even and 
\begin{eqnarray*}
x & = & \al/2-b                  \\
y & = & d_2-d_1-a+\al/2 \\ 
z & = & d_1-\al/2             \\
w & = & d_B-d_1-a-b+\al/2 
\end{eqnarray*}
\end{lemma}
\begin{proof}
Note that (see Fig.~\ref{f:l-l+1}(a)):
\[
x+z=d_1-b \ \ \ \ \ y+z=d_2-a \ \ \ \ \ x+y=d_A-a-b \ \ \ \ \ z+w=d_B-a-b  
\]
Solving this equation system for $x,y,z,w$ gives the lemma. 
\end{proof}

When both $A,B$ are $\la$-cuts, then it is well known that $\la$ must be even, the square has no diagonal edges,
and each side edge has capacity $\la/2$; c.f. \cite{DKL}. The following is also known.

\begin{lemma} [\cite{DN,N-th}] \label{l:l-l+1}
Suppose that $A$ is a $\la$-cut and $B$ is a $(\la+1)$-cut. 
Then the square has no diagonal edges, and the following holds.
\begin{enumerate}[(i)] 
\item 
If $\la$ is odd then the square has one side edge of capacity $(\la-1)/2$ and other three side edges have capacity $(\la+1)/2$; 
see Fig.{\em ~\ref{f:l-l+1}(b)}. 
\item
If $\la$ is even then the square has one side edge of capacity $\la/2+1$ and other three side edges have capacity $\la/2$; 
see Fig.~{\em \ref{f:l-l+1}(c)}. 
\end{enumerate}
\end{lemma}
\begin{proof}
If $\la$ is odd, $d_1,d_2$ have distinct parity, since $\al=\la+d_1-d_2$ must be even.
By (\ref{e:a},\ref{e:b}) $(d_1,d_2)=(\la,\la+1)$ and $(a,b)=(0,0)$. Then $\al/2=(\la-1)/2$ and from Lemma~\ref{l:sol} we get
$(x,y,z,w)=\left(\f{\la-1}{2},\f{\la+1}{2},\f{\la+1}{2},\f{\la+1}{2}\right)$, which is the case in Fig.~\ref{f:l+1}(b).

If $\la$ is even, $d_1,d_2$ have the same parity.
By (\ref{e:a},\ref{e:b}) $(d_1,d_2)=(\la,\la)$, $(a,b)=(0,0)$, and $\al/2=\la/2$. 
Then $(x,y,z,w)=\left(\f{\la}{2},\f{\la}{2},\f{\la}{2},\f{\la}{2}+1\right)$ by Lemma~\ref{l:sol},
which is the case in Fig.~\ref{f:l+1}(c).
\end{proof}

The next lemma deals with two $(\la+1)$-cuts.

\begin{lemma} \label{l:l+1}
Suppose that $A,B$ are $(\la+1)$-cuts. 
If $\la$ is even then one of the following holds.
\begin{enumerate}[(a)]
\item
The square has no diagonal edges and has two adjacent side edges of capacity $\la/2$ 
while the other two edges have capacity $\la/2+1$; 
see Fig.~{\em \ref{f:l+1}(a)}.
\item
The square has one diagonal edge and all side edges have capacity $\la/2$; 
see Fig.~{\em \ref{f:l+1}(b)}. 
\end{enumerate}

If $\la$ is odd then one of the following holds.
\begin{enumerate}[(a)]\setcounter{enumi}{2}
\item
The square has no diagonals, two opposite side edges have capacity $(\la+1)/2$,
one side edge has capacity $(\la+3)/2$ and its opposite side edge has capacity $(\la-1)/2$; 
see Fig.~{\em \ref{f:l+1}(c)}.
\item
The square has one diagonal edge, 
two side edge incident to one end of the the diagonal edge have capacity $(\la-1)/2$, 
while the other have capacity $(\la+1)/2$; 
see Fig.~{\em \ref{f:l+1}(d)}.
\item
The square has both diagonal edges of capacity $1$ each, 
and all side edges have capacity $(\la-1)/2$;
see Fig.{\em ~\ref{f:l+1}(e)}.
\item
The square has no diagonal edges and all side edges have capacity $(\la+1)/2$;
see Fig.~{\em \ref{f:l+1}(f)}.
\end{enumerate}
\end{lemma}
\begin{proof}
We apply Lemma~\ref{l:sol}. 
Suppose that $\la$ is even. Since $\al=\la+1+d_1-d_2$ must be even, $d_1,d_2$ have distinct parity.
By (\ref{e:a},\ref{e:b}) $(d_1,d_2)=(\la,\la+1)$ and $b=0$. Then $\al/2=\la/2$ and we have two cases. 
\begin{enumerate}[(a)]
\item
$(a,b)=(0,0)$. Then $(x,y,z,w)=\left(\f{\la}{2},\f{\la}{2}+1,\f{\la}{2},\f{\la}{2}+1\right)$,
which is the case in Fig.~\ref{f:l+1}(a).
\item
$(a,b)=(1,0)$. Then $(x,y,z,w)=\left(\f{\la}{2},\f{\la}{2},\f{\la}{2},\f{\la}{2}\right)$,
which is the case in Fig.~\ref{f:l+1}(b).
\end{enumerate}
If $\la$ is odd, $d_1,d_2$ have the same parity.
By (\ref{e:a},\ref{e:b}) $(d_1,d_2)=(\la,\la)$ or $(d_1,d_2)=(\la+1,\la+1)$; in both cases $\al/2=(\la+1)/2$. \\ 
For $(d_1,d_2)=(\la,\la)$ and $\al/2=(\la+1)/2$ we have the following cases.
\begin{enumerate}[(a)]\setcounter{enumi}{2}
\item 
$(a,b)=(0,0)$. Then $(x,y,z,w)=\left(\f{\la+1}{2},\f{\la+1}{2},\f{\la-1}{2},\f{\la+3}{2}\right)$,
which is the case in Fig.~\ref{f:l+1}(c).
\item 
$(a,b)=(1,0)$. Then $(x,y,z,w)=\left(\f{\la+1}{2},\f{\la-1}{2},\f{\la-1}{2},\f{\la+1}{2}\right)$,
which is the case in Fig.~\ref{f:l+1}(d).
\item
$(a,b)=(1,1)$. Then $(x,y,z,w)=\left(\f{\la-1}{2},\f{\la-1}{2},\f{\la-1}{2},\f{\la-1}{2}\right)$,
which is the case in Fig.~\ref{f:l+1}(e).
\end{enumerate}
For $(d_1,d_2)=(\la+1,\la+1)$ we have one case
\begin{enumerate}[(a)]\setcounter{enumi}{5}
\item
$(a,b)=(0,0)$. Then $(x,y,z,w)=(\f{\la+1}{2}, \f{\la+1}{2},\f{\la+1}{2},\f{\la+1}{2})$, which is the case in Fig.~\ref{f:l+1}(f).
\end{enumerate}

This concludes the proof of the lemma.
\end{proof}

\begin{figure} \centering \includegraphics{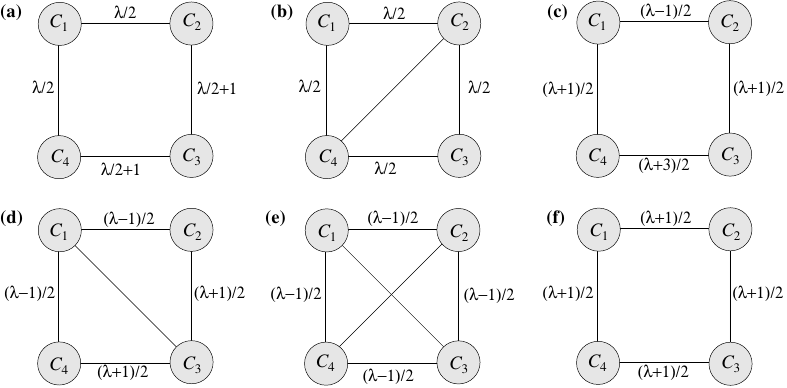}
\caption{Illustration to Lemma \ref{l:l+1}.}
\label{f:l+1} \end{figure}

From Lemma~\ref{l:l+1} we get the following.

\begin{corollary} \label{c:la}
Let $H$ be a $\la$-edge connected graph and let $\FF=\{S \subs V: d_H(S) \in \{\la,\la+1\}\}$.
If $\la$ is even then $\FF$ is uncrossable. 
\end{corollary}
\begin{proof}
Let $A,B \in \FF$. We need to show that $A \cap B,A \cup B \in \FF$ or $A \sem B,B \sem A \in \FF$. 
If one of $A \sem B,B \sem A$                 is empty then $\{A \cap B,A \cup B\}=\{A,B\}$. 
If one of $A \cap B, V \sem (A \cup B)$  is empty then $\{A \sem B,B \sem A\}=\{A,B\}$. 
In both cases, the lemma holds. 

If all corner four sets are non-empty then one of the cases (a,b) in Lemma~\ref{l:l+1} holds. 
In case (a) $A \sem B,B \sem A \in \FF$ (see Fig.~\ref{f:l+1}(a)), and 
in case (b) all corner sets are in $\FF$ (see Fig.~\ref{f:l+1}(b)).
\end{proof}

Let $H$ be a graph and $\FF$ a family of node subsets (or cuts) of $H$. 
It is easy to see that the relation $\{(u,v) \in V \times V: \mbox{no member of } \FF \mbox{ separates } u,v\}$
is an equivalence; the {\bf quotient graph} of $\FF$ is obtained by shrinking each equivalence class of this relation into a single node,
and replacing every set of parallel edges by a single capacitated edge. 
$\FF$ is {\bf symmetric} if $S \in \FF$ implies $V \sem S \in \FF$.
We say that $\FF$ is a {\bf crossing family} if $A \cap B,A \cup B \in \FF$ for any $A,B$ that cross,
and if in addition $(A \sem B) \cup (B \sem A) \notin \FF$ then $\FF$ is a {\bf proper symmetric crossing family}.
By \cite{DN,N-th}, the problem of covering a symmetric crossing family is equivalent to covering the minimum cuts 
of a $2$-edge connected graph, while \cite{TZ2} shows that the later problem admits a $(1.5+\eps)$-approximation algorithm.

We need the following result from \cite{DN,N-th} (specifically, see Lemmas 3.11--3.15 in \cite{N-th}.) 

\begin{figure} \centering \includegraphics{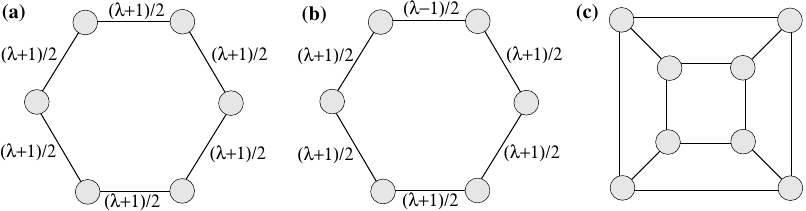}
\caption{Illustration to Lemma~\ref{l:DN}.}
\label{f:cycle} \end{figure}

\begin{lemma} [\cite{DN,N-th}] \label{l:DN}
Let $H$ be a $\la$-edge-connected graph with $\la$ odd, and let $\FF$ be the family of $(\la+1)$-cuts of $H$. 
Then there exists a subfamily $\FF'$ of the $\la$-cuts of $H$ such that 
$\FF \cup \FF'$ can be decomposed in polynomial time into parts whose union contains $\FF$, 
such that every cut in $\FF$ belongs to at most $2$ parts, and such that 
the cuts in each part correspond to the $(\la+1)$-cuts of its quotient graph, which is either (see Fig.~\ref{f:cycle}):
\begin{enumerate}[(a)]
\item
A cycle of edges of capacity $(\la+1)/2$ each.
\item
A cycle with one edge of capacity $(\la-1)/2$ and other edges of capacity $(\la+3)/2$ each. 
\item
A cube graph, which can occur only if $\la=3$.
\end{enumerate}
\end{lemma}

\section{Algorithm for {\nmcc} (Theorem~\ref{t:la})} \label{s:la}

It is known that the problem of finding a minimum cost cover of cuts of size $<k$ 
of a $\la$-connected graph admits the following approximation ratios for $k \leq \la+2$:
\begin{itemize}
\item
$1.5+\eps$ if $k=\la+1$ \cite{TZ,TZ2}.
\item
$2$ if $k=\la+2$ and $\la$ is even by Corollary~\ref{c:la}, 
since the problem of covering an uncrossable family admits ratio $2$ \cite{GGPS}. 
\end{itemize}

We will show that this implies the approximation ratios for {\nmcc} claimed in Theorem~\ref{t:la}:  
\begin{itemize}
\item
$k-\la_0$                   \hspace*{1.55cm} if $\la_0,k$ are both even. 
\item
$k-\la_0+1/2+\eps$ \                                if $\la_0,k$ have distinct parity.
\item
$k-\la_0+1+\eps$    \hspace*{0.38cm}  if $\la_0,k$ are both odd.
\end{itemize}

Algorithm~\ref{alg:even} computes a solution as required when $\la_0,k$ are both even. 
\begin{algorithm} \caption{$\la_0,k$ are both even} 
\label{alg:even}
$F \gets \empt$ \\
\For{\em $\la=\la_0$ to $\la=k-2$}
{
Find a $2$-approximate cover $J$ of $\{\la,\la+1\}$-cuts of  $G_0 \cup F$ where edges in $F$ have capacity/multiplicity $k-\la_0$. \\
$F \gets F \cup J$.
}
\Return{$F$}
\end{algorithm}

The correctness of the algorithm is straightforward. The approximation ratio $k-\la_0$ follows from the observation 
that we pay $2 \cdot {\sf opt}$ at each iteration, so ${\sf opt}$ per increasing the connectivity by $1$. 
If $\la_0$ is odd and $k$ is even, then in the first iteration we apply the $(1.5+\eps)$-approximation algorithm of \cite{TZ} 
for covering $\la$-cuts only and have an extra $1/2+\eps$ term.
If $\la_0$ is even and $k$ is odd, then in the last iteration we apply a $2$-approximation algorithm 
for covering $(k-1)$-cuts only, and also have an extra $1/2+\eps$ term \cite{TZ2}. 
If both of these occur then the extra term is $1+2\eps \approx 1+\eps$. 

This concludes the proof of Theorem~\ref{t:la}. 

\section{Algorithm for ({\em k,q})-{\fgc} (Theorem~\ref{t:fgc})} \label{s:fgc}

Let $\langle G=(V,E), U, k \rangle$ be an instance of $(k,q)$-{\fgc}. 
Let $H$ be a subgraph of $G$.
 We will use the notation $d(S)=d_H(S)$ and $d_U(S)=d_{H \cap U}(S)$. 
Recall that a subgraph $H$ of $G$ is $(k,q)$-flex-connected if (\ref{e:flex}) holds, namely if:
\[
d(S) \geq k+\min\{d_U(S),q\} \ \ \ \ \ \forall \ \empt \neq S \subsetneq V
\]
Suppose that $H$ is $(k,q-1)$-flex-connected. Then 
\[
d(S) \geq k+\min\{d_U(S),q-1\} \ \ \ \ \ \forall \ \empt \neq S \subsetneq V
\]
One can see that if for $\empt \neq S \subsetneq V$ the later inequality holds but not the former then  
$d(S)=k+q-1$ and $d_U(S) \geq q$. 
Consequently, to make $H$ $(k,q)$-flex-connected we need to cover the set family 
\[
\FF_q(H)=\{\empt \neq S \subsetneq V:d(S)=k+q-1, d_U(S) \geq q\} \ .
\]
Thus the following algorithm computes a feasible solution for $(k,q)$-{\fgc}, see also \cite{CJ}.
\begin{algorithm} \caption{\sc Iterative-Cover$(G,c,U,k,q)$} 
\label{alg:fgc}
Compute a $k$-edge-connected spanning subgraph $H=(V,J)$ of $G$. \\
\For{\em $\ell=1$ to $q$}
{Add to $H$ a cover $J_\ell$ of $\FF_\ell(H)$.
}
\Return{$H$}.
\end{algorithm}

We can use this observation to prove parts (i) and (iii) of Theorem~\ref{t:fgc}: 
that $(k,1)$-{\fgc} admits ratio $3.5+\eps$ if $k$ is odd, and that $(k,q)$-{\fgc} admits ratio $\al+\f{2q}{k}$ for unit costs,
where $\al$ is the best known ratio for the {\sc Min-Size $k$-Edge-Connected Subgraph} problem.

If $k$ is odd then it is known that the family $\{\empt \neq S \subsetneq V: d_H(S)=k\}$ is laminar, 
and thus any its subfamily, and in particular $\FF_1(H)$, is also laminar.
Part (i) now follows from the fact that the problem of covering a laminar family admits ratio $1.5+\eps$ \cite{TZ}. 

For part (iii) we need the known fact that any inclusion minimal cover $J$ of a set family $\FF$ is a forest.
To see this, suppose to the contrary that $J$ contains a cycle $C$. 
Let $e=uv$ be an edge of $C$.
Since $P = C \sem \{e\}$ is a $uv$-path,
then for any $A \in \FF$ covered by $e$, there is $e' \in P$ that covers $A$. 
This implies that $J \sem \{e\}$ also covers $\FF$, contradicting the minimality of $J$.

On the other hand, ${\sf opt} \geq kn/2$, hence $|J_i|/{\sf opt} \leq 2(n-1)/kn <2/k$. 
Thus the overall approximation ratio is $\al+2q/k$, as claimed. 
It remains to show that the algorithm can be implemented in polynomial time.
By \cite{BCHI} (see also \cite[Lemma~2.8]{CJ}), if $H$ is $(k,i-1)$-flex-connected then $|\FF_i(H)|=O(n^4)$, 
and the members of $\FF_i(H)$ can be listed in polynomial time.
Since a $(k,i-1)$-flex-connected $H$ is $(k,i)$-flex-connected if and only if $\FF_i(H)=\empt$, 
we can compute at iteration $i$ an inclusion minimal cover of $\FF_i(H)$ in polynomial time 

\medskip

The proof of part (ii) of Theorem~\ref{t:fgc} relies on several lemmas.

\begin{lemma} \label{l:C1}
Let $H$ be $(k,q-1)$-flex-connected, $q \geq 2$, and let $A,B \in \FF_q(H)$ cross. 
Then:
\begin{enumerate}[(i)]
\item
If $d_U(C_1) \geq 1$ then $d(C_1)+d(C_2) \geq 2k+q$.
\item
If $d_U(C_1)=0$ then $C_2,C_4 \in \FF_q(H)$. 
\end{enumerate}
\end{lemma}
\begin{proof}
Suppose that $d_U(C_1) \geq 1$. Note that 
$$
d_U(C_1)+d_U(C_2) \geq d_U(A) \geq  q \ .
$$
Since $H$ is $(k,q-1)$-flex-connected  
$$
d(C_1)+d(C_2)-2k \geq \min\{d_U(C_1),q-1\}+\min\{d_U(C_2),q-1\} \ .
$$
If $d_U(C_2) \geq 1$ and $q \geq 2$ then the r.h.s. is at least
$$
\min\{\min\{d_U(C_1),d_U(C_2)\}+q-1,d_U(C_1)+d_U(C_2),2q-2\}  \geq \min\{q,q,2q-2\}=q
$$
If $d_U(C_2) =0$ then $d_U(C_1) \geq q$, hence $d(C_1) \geq k+q-1$. 
Since $d(C_2) \geq d(C_1)$, we get $d(C_1) +d(C_2) \geq 2(k+q-1) \geq 2k+q$.

Suppose that $d_U(C_1)=0$. Then $d_U(C_2)=d_U(A)\geq q$ and $d_U(C_4)=d_U(B) \geq q$.
Since $H$ is $(k,q-1)$-flex-connected, $d(C_2) \geq k+q-1$ and $d(C_4) \geq k+q-1$.
Therefore 
$$
2(k+q-1) \leq d(C_2)+d(C_4) \leq d(A)+d(B) \leq 2(k+q-1)
$$ 
Hence equality holds everywhere, hence $d(C_2)=d(C_4) = k+q-1$, so $C_2,C_4 \in \FF_q(H)$. 
\end{proof}

\begin{lemma} \label{l:even}
Let $k$ be even. If $H$ is $(k,1)$-flex-connected then $\FF_2(H)$ is uncrossable.
\end{lemma}
\begin{proof}
Let $A,B \in \FF_2(H)$ cross. By Lemma~\ref{l:l+1}(a,b), $d(C_1)+d(C_2)=2k+1$.
Thus we must be in case~(ii) of Lemma~\ref{l:C1}, implying $C_2,C_4 \in \FF_q(H)$.  
\end{proof}

We need some definitions to handle the case $q=2$ and $k$ odd. 
Two sets $A,B \subset V$ {\bf cross} if  
$A \cap B, V \sem (A \cup B)$ are non-empty.
Let $\FF$ be a set family. 
We say that $\FF$ is a {\bf crossing family} if $A \cap B,A \cup B \in \FF$ for any $A,B$ that cross, and
if in addition $(A \sem B) \cup (B \sem A) \notin \FF$ then $\FF$ is a {\bf proper crossing family}.
$\FF$ is {\bf symmetric} if $S \in \FF$ implies $V \sem S \in \FF$.
By \cite{DN,N-th}, the problem of covering a symmetric proper crossing family is equivalent to covering the minimum cuts 
of a $2$-edge connected graph, while \cite{TZ2} shows that the later problem admits a $(1.5+\eps)$-approximation algorithm.
As was mentioned, $\FF_1(H)$ is laminar if $k$ is odd, 
and it is also not hard to see that $\FF_1(H)$ is uncrossable if $k$ is even, see \cite{BCHI}. 
The following lemma gives uncrossing properties of $\FF_2(H)$, assuming that $\FF_1(H)=\empt$.

\begin{lemma} \label{l:odd}
Let $k$ be odd. 
If $H$ is $(k,1)$-flex-connected then $\FF_2(H)$ can be decomposed in polynomial time into 
an uncrossable family $\FF'$ and a symmetric proper crossing family $\FF''$ such that $\FF' \cup \FF''=\FF_2(H)$. 
\end{lemma}
\begin{proof}
Consider a decomposition of $(k+1)$-cuts as in Lemma~\ref{l:DN}.
Note that if two $(k+1)$-cuts belong to the same part then so are their corner cuts. 
Hence to prove the lemma it is sufficient to provide a proof for each part of $\FF$ as in Lemma~\ref{l:DN},
namely, we may assume that there is just one part which is $\FF$. 

For an edge $e$ of the quotient graph of $\FF$ 
let $u_e$ be the number of unsafe edges in the edge subset of $H$ represented by $e$. 
Let us say that $e$ is {\bf red} if $u_e \geq 2$, {\bf blue} if $u_e=1$, and {\bf black} otherwise (if $u_e=0$). 
Note that since $\FF_1(H)=\empt$, there cannot be a $k$-cut in the quotient graph that contain a blue or a red edge.
Consequently, only in case (a) the quotient graph may have non black edges, as in cases (b,c) every edge of the quotient graph
belongs to some $k$-cut. 
On the other hand, every cut in $\FF_2(H)$must contain a blue or a red edge, thus the only relevant case is (a).

Let $\FF'$ be the family of cuts in $\FF_2$ that contain a red edge and let $\FF''=\FF \sem \FF'$.
Note that every cut in $\FF''$ consists of $2$ blue edges. 
We claim that $\FF'$ is uncrossable and that $\FF''$ is is a proper crossing family (clearly, $\FF''$ is symmetric).
Let $A,B \in \FF$ cross. 
If $A,B \in \FF'$ then their square has two adjacent red edges, and this implies that 
$A \cup B,A \cap B \in \FF$ or $A \sem B,B \sem A \in \FF'$. Consequently, $\FF'$ is uncrossable.
If $A,B \in \FF''$ then their square has $4$ blue edges, and then all corner cuts are in $\FF''$. Thus $\FF''$ is a crossing family.
Furthermore, the capacity of the cut defined by the set $(A \sem B) \cup (B \sem A)$ is $4(k+1)/2=2(k+1)>k+1$,
and thus $\FF''$ is a proper crossing family.
\end{proof}

We now finish the proof of part (ii) of Theorem~\ref{t:fgc}. 
We will apply Algorithm~\ref{alg:fgc}.
At step~1, $c(J) \leq 2{\sf opt}$, c.f. \cite{K}. 
In the loop (steps 2,3) we combine Lemmas \ref{l:even} and \ref{l:odd}  
with the best known approximation ratios for solving appropriate set family edge cover problems. 

If $k$ is even then each of $\FF_1,\FF_2$ is uncrossable by Lemma~\ref{l:even} and thus 
$c(J_1),c(J_2) \leq 2{\sf opt}$, by \cite{GGPS}; 
consequently, $c(J \cup J_1 \cup J_2) \leq 6{\sf opt}$. 

If $k$ is odd then $\FF_1$ is laminar and and thus $c(J_1) \leq (1.5+\eps){\sf opt}$ by \cite{TZ}.
After $\FF_1$ is covered, $\FF_2$ can be decomposed  into 
an uncrossable family $\FF'$ and a symmetric proper crossing family $\FF''$, by Lemma~\ref{l:odd}.
We can compute a $2$-approximate cover $J'$ of $\FF'$ using the algorithm of \cite{GGPS} and a
$(1.5+\eps)$-approximate cover $J''$ of $\FF''$ using the algorithm of \cite{TZ2}. 
Consequently, the approximation ratio of the algorithm is bounded by
\[
[c(J)+c(J_1)+c(J')+c(J'')]/{\sf opt}  \leq 2+(1.5+\eps)+2+(1.5+\eps) =7+2\eps \approx 7+\eps  \ ,
\]
concluding the proof of part (ii) of Theorem~\ref{t:fgc}.


\begin{thebibliography}{10}

\bibitem{AHM}
D.~Adjiashvili, F.~Hommelsheim, and M.~M\"{u}hlenthaler.
\newblock Flexible graph connectivity.
\newblock {\em Mathematical Programming}, pages 1--33, 2021.

\bibitem{BCGI}
I.~Bansal, J.~Cheriyan, L.~Grout, and S.~Ibrahimpur.
\newblock Improved approximation algorithms by generalizing the primal-dual
  method beyond uncrossable functions.
\newblock {\em CoRR}, abs/2209.11209, 2022.
\newblock URL: \url{https://arxiv.org/abs/2209.11209}.

\bibitem{BCHI}
S.~C. Boyd, J.~Cheriyan, A.~Haddadan, and S.~Ibrahimpur.
\newblock Approximation algorithms for flexible graph connectivity.
\newblock In {\em FSTTCS}, page 9:1–9:14, 2021.

\bibitem{CJ}
C.~Chekuri and R.~Jain.
\newblock Augmentation based approximation algorithms for flexible network
  design.
\newblock {\em CoRR}, abs/2209.12273, 2022.
\newblock URL: \url{https://doi.org/10.48550/arXiv.2209.12273}.

\bibitem{DKL}
E.~A. Dinic, A.~V. Karzanov, and M.~V. Lomonosov.
\newblock On the structure of the system of minimum edge cuts in a graph.
\newblock In A.~A. Fridman, editor, {\em Studies in Dtscrete Optimization},
  pages 290--306. Nauka, Moscow, 1976.
\newblock in Russian.

\bibitem{DN}
Y.~Dinitz and Z.~Nutov.
\newblock A 2-level cactus model for the system of minimum and minimum+1
  edge-cuts in a graph and its incremental maintenance.
\newblock In {\em STOC}, pages 509--518, 1995.

\bibitem{GG}
H.~N. Gabow and S.~Gallagher.
\newblock Iterated rounding algorithms for the smallest $k$-edge connected
  spanning subgraph.
\newblock {\em SIAM J. Computing}, 41(1):61--103, 2012.

\bibitem{GGTW}
H.~N. Gabow, M.~X. Goemans, \'{E}. Tardos, and D.~P. Williamson.
\newblock Approximating the smallest $k$-edge connected spanning subgraph by
  {LP}-rounding.
\newblock {\em Networks}, 53(4):345--357, 2009.

\bibitem{GGJ}
M.~Garg, F.~Grandoni, and A.~Jabal Ameli.
\newblock Improved approximation for two-edgeconnectivity.
\newblock {\em CoRR}, abs/2209.10265v2, 2022.
\newblock URL: \url{https://arxiv.org/abs/2209.10265v2}.

\bibitem{GGPS}
M.~X. Goemans, A.~V. Goldberg, S.~A. Plotkin, D.~B. Shmoys, \'{E}. Tardos, and
  D.~P. Williamson.
\newblock Improved approximation algorithms for network design problems.
\newblock In {\em SODA}, pages 223--232, 1994.

\bibitem{K}
S.~Khuller.
\newblock Approximation algorithms for finding highly connected subgraphs.
\newblock In D.~Hochbaum, editor, {\em Approximation Algorithms for NP-hard
  problems}, chapter~6, pages 236--265. PWS, 1995.

\bibitem{N-th}
Z.~Nutov.
\newblock {\em Structures of Cuts and Cycles in Graphs; Algorithms and
  Applications Algorithms and Applications}.
\newblock PhD thesis, Technion, Israel Institute of Technology, 1997.

\bibitem{TZ2}
V.~Traub and R.~Zenklusen.
\newblock A (1.5+{\(\epsilon\)})-approximation algorithm for weighted
  connectivity augmentation.
\newblock {\em CoRR}, abs/2209.07860, 2022.
\newblock URL: \url{https://doi.org/10.48550/arXiv.2209.07860}, \href
  {http://arxiv.org/abs/2209.07860} {\path{arXiv:2209.07860}}.

\bibitem{TZ}
V.~Traub and R.~Zenklusen.
\newblock Local search for weighted tree augmentation and steiner tree.
\newblock In {\em SODA}, pages 3253--3272, 2022.

\end{thebibliography}

\end{document}